\documentclass[11pt]{article}

\input{preamble}

\DeclareMathOperator{\dist}{dist}
\DeclareMathOperator{\spc}{\textsc{spc}}
\newcommand{\opt}{\rho}

\title{Fixed-Parameter Approximations for $k$-Center Problems\\in Low 
Highway Dimension Graphs\thanks{I would like to thank Jochen K{\"o}nemann for 
reading an early draft of this paper. Also I would like to thank two anonymous 
reviewers for their insightful remarks that helped to greatly improve the 
paper. A preliminary version appeared at the 42nd International Colloquium on 
Automata, Languages, and Programming (ICALP 2015). This work was supported by 
ERC Starting Grant PARAMTIGHT (No.~280152), project CE-ITI (GA\v{C}R 
no.~P202/12/G061) of the Czech Science Foundation, and by the Center for 
Foundations of Modern Computer Science (Charles Univ.\ project UNCE/SCI/004).}}
\author{Andreas Emil Feldmann}
\affil{KAM, Charles University, Czechia\\
\texttt{feldmann.a.e@gmail.com}}

\date{}

\begin{document}
\renewcommand*{\sectionautorefname}{Section}
\renewcommand*{\subsectionautorefname}{Section}
\renewcommand*{\algorithmautorefname}{Algorithm}

\maketitle

\begin{abstract}
We consider the \pname{$k$-Center} problem and some generalizations. For 
\pname{$k$-Center} a set of $k$ \emph{center vertices} needs to be found in a 
graph $G$ with edge lengths, such that the distance from any vertex of~$G$ to 
its nearest center is minimized. This problem naturally occurs in transportation 
networks, and therefore we model the inputs as graphs with bounded \emph{highway 
dimension}, as proposed by Abraham et~al.~[SODA~2010]. 

We show both approximation and fixed-parameter hardness results, and how to 
overcome them using \emph{fixed-parameter approximations}, where the two 
paradigms are combined. In particular, we prove that for any $\eps>0$ computing 
a $(2-\eps)$-approximation is W[2]-hard for parameter~$k$, and NP-hard for 
graphs with highway dimension $O(\log^2 n)$. The latter does not rule out 
fixed-parameter $(2-\eps)$-approximations for the highway dimension 
parameter~$h$, but implies that such an algorithm must have at least doubly 
exponential running time in $h$ if it exists, unless ETH fails. On the positive 
side, we show how to get below the approximation factor of~$2$ by combining the 
parameters $k$ and~$h$: we develop a fixed-parameter $3/2$-approximation with 
running time $2^{O(kh\log h)}\cdot n^{O(1)}$. Additionally we prove that, unless 
P=NP, our techniques cannot be used to compute fixed-parameter 
$(2-\eps)$-approximations for only the parameter $h$. 

We also provide similar fixed-parameter approximations for the \pname{weighted 
$k$-Center} and \pname{$(k,\mc{F})$-Partition} problems, which generalize 
\pname{$k$-Center}.
\end{abstract}

\section{Introduction}

In this paper we consider the \pname{$k$-Center} problem and some of its 
generalizations. For the problem, $k$ locations need to be found in a network, 
so that every node in the network is close to a location. More formally, the 
input is specified by an integer $k\in\mathbb{N}$ and a graph $G=(V,E)$ with 
positive edge lengths. A \emph{feasible solution} to the problem is a set 
$C\subseteq V$ of \emph{centers} such that~$|C|\leq k$. The aim is to minimize 
the maximum distance between any vertex and its closest center. That is, let 
$\dist_G(u,v)$ denote the shortest-path distance between two vertices $u,v\in V$ 
of $G$ according to the edge lengths, and $B_v(r)=\{u\in V\mid \dist_G(u,v)\leq 
r\}$ be the \emph{ball of radius $r$ around $v$}. We need to minimize the 
\emph{cost} of the solution $C$, which is the smallest value~$\opt$ for which 
$\bigcup_{v\in C} B_v(\opt)= V$. We say that a center $v\in C$ \emph{covers} a 
vertex $u\in V$ if $u\in B_v(\rho)$. Hence we can see the problem as finding $k$ 
centers covering all vertices of $G$ with balls of minimum radius.

The \pname{$k$-Center} problem naturally arises in transportation networks, 
where, for instance, it models the need to find locations for manufacturing 
plants, hospitals, police stations, or warehouses under a budget constraint. 
Unfortunately it is NP-hard to solve the problem in 
general~\cite{Vazirani01book}, and the same holds true in various models for 
transportation networks, such as planar graphs~\cite{plesnik1980planar-k-center} 
and metrics using Euclidean~($L_2$), Manhattan~($L_1$), or 
Chebyshev~($L_\infty$) distance measures~\cite{feder1988metric-k-center}. A more 
recent model for transportation networks uses the \emph{highway dimension}, 
which was introduced as a graph parameter by \citet{abraham2010highway}. The 
intuition behind its definition comes from the empirical 
observation~\cite{bast2007transit,bast2009ultrafast} that in a road network, 
starting from any point $A$ and travelling to a sufficiently far point $B$ along 
the quickest route, one is bound to pass through some member of a sparse set of 
``access points''. There are several formal definitions for the highway 
dimension that differ 
slightly~\cite{abraham2010highway,abraham2010highway2,abraham2011vc, 
FeldmannFKP-highway-2015}. All of them, however, imply the existence of 
\emph{locally sparse shortest path covers}. Therefore, in this paper we consider 
this as a generalization of the original highway dimension definitions.

\begin{dfn}\label{dfn:spc}
Given a graph $G=(V,E)$ with edge lengths and a \emph{scale} $r\in\mathbb{R}^+$, 
let $\mc{P}_{(r,2r]}\subseteq 2^V$ contain all vertex sets given by shortest 
paths in $G$ of length more than $r$ and at most~$2r$. A \emph{shortest path 
cover} \mbox{$\spc(r)\subseteq V$} is a hitting set for the set system 
$\mc{P}_{(r,2r]}$, i.e., $P\cap \spc(r)\neq\emptyset$ for each 
$P\in\mc{P}_{(r,2r]}$. We call the vertices in $\spc(r)$ \emph{hubs}. A hub set 
$\spc(r)$ is called {\em locally $h$-sparse} if for every vertex $v\in V$ the 
ball $B_v(2r)$ of radius $2r$ around $v$ contains at most $h$ vertices from 
$\spc(r)$. The \emph{highway dimension} of $G$ is the smallest integer $h$ such 
that there is a locally $h$-sparse shortest path cover $\spc(r)$ for every scale 
$r\in\mathbb{R}^+$ in $G$.
\end{dfn}

\citet{abraham2010highway} introduced the highway dimension in order to explain 
the fast running times of various shortest-path heuristics. However, they also 
note that ``conceivably, better algorithms for other [optimization] problems can 
be developed and analysed under the small highway dimension assumption''. In 
this paper we investigate the \pname{$k$-Center} problem and focus on graphs 
with low highway dimension as a model for transportation networks. One advantage 
of using such graphs is that they do not only capture road networks but also 
networks with transportation links given by air-traffic or railroads. For 
instance, introducing connections due to airplane traffic will render a network 
non-planar, while it can still be argued to have low highway dimension: longer 
flight connections tend to be served by bigger but sparser airports, which act 
as hubs. This can, for instance, be of interest in applications where warehouses 
need to be placed to store and redistribute goods of globally operating 
enterprises. Unfortunately however, in this paper we show that the 
\pname{$k$-Center} problem also remains NP-hard on graphs with low highway 
dimension.

Two popular and well-studied ways of coping with NP-hard problems is to devise 
\emph{approximation}~\cite{Vazirani01book,williamson2011design} and 
\emph{parameterized}~\cite{downey2013fpt,pc-book} algorithms. For the former we 
demand polynomial running times but allow the computed solution to deviate from 
the optimum cost. That is, we compute a \emph{$c$-approximation}, which is a 
feasible solution with a cost that is at most $c$ times worse than the best 
possible for the given instance. A problem that allows a polynomial-time 
$c$-approximation for any input is \emph{$c$-approximable}, and $c$ is called 
the \emph{approximation factor} of the corresponding algorithm. The rationale 
behind parameterized algorithms is that some \emph{parameter} $p$ of the input 
is small and we can therefore afford running times that are super-polynomial in 
$p$, while, however, we demand optimum solutions. That is, we compute a solution 
with optimum cost in time $f(p)\cdot n^{O(1)}$ for some computable function 
$f(\cdot)$ that is independent of the input size $n$. A problem that has a 
fixed-parameter algorithm for a parameter $p$ is called \emph{fixed-parameter 
tractable~(FPT)} for $p$. What however, if a problem is neither approximable nor 
FPT? In this case it may be possible to overcome the complexity by combining 
these two paradigms. In particular, the objective becomes to develop 
\emph{fixed-parameter $c$-approximation} ($c$-FPA) algorithms that compute a 
$c$-approximation in time $f(p)\cdot n^{O(1)}$ for a parameter~$p$. 

The idea of combining the paradigms of approximation and fixed-parameter 
tractability has been suggested before. However, only few results are known for 
this setting (cf.~\cite{marx2008fpa}). In this paper we show that for the 
\pname{$k$-Center} problem it is possible to overcome lower bounds for its 
approximability and its fixed-parameter tractability using parameterized 
approximations. For many different input classes, such as planar 
graphs~\cite{plesnik1980planar-k-center}, and $L_1$- and 
$L_\infty$-metrics~\cite{feder1988metric-k-center}, the \pname{$k$-Center} 
problem is $2$-approximable via the algorithm for general metrics of 
\citet{hochbaum1986bottleneck}, but not $(2-\eps)$-approximable for 
any~$\eps>0$, unless P=NP. We show that, unless FPT=W[2], for general graphs 
there is no \mbox{$(2-\eps)$-FPA} algorithm for the parameter $k$. Additionally, 
we prove that, unless P=NP, \pname{$k$-Center} is not $(2-\eps)$-approximable on 
graphs with highway dimension $O(\log^2 n)$. This does not rule out 
$(2-\eps)$-FPA algorithms for the highway dimension parameter, and we leave this 
as an open problem. However, the result implies that if such an algorithm 
exists, then its running time must be enormous. In particular, unless the 
exponential time hypothesis (ETH)~\cite{eth,eth-2} fails, there can be no 
$(2-\eps)$-FPA algorithm with doubly exponential $2^{2^{o(\sqrt{h})}}\cdot 
n^{O(1)}$ running time in the highway dimension~$h$.

In face of these hardness results, it seems tough to beat the approximation 
factor of $2$ for \pname{$k$-Center}, even when considering fixed-parameter 
approximations for either the parameter $k$ or the highway dimension. Our main 
result, however, is that we can obtain a significantly better approximation 
factor for \pname{$k$-Center} when combining these two parameters. Such an 
algorithm is useful when aiming for high quality solutions, for instance, in a 
setting where only few warehouses should be built in a transportation network, 
since warehouses are expensive or stored goods should not be too dispersed for 
logistical reasons. 

It is known~\cite{abraham2011vc} that locally $O(h\log h)$-sparse shortest path 
covers can be computed for graphs of highway dimension $h$ in polynomial time, 
if each shortest path is unique. We will assume that the latter is always the 
case, since we can slightly perturb the edge lengths. In particular, using a 
folklore method we may distort distances such that any $3/2$-approximation in 
the perturbed instance also is a $3/2$-approximation in the original instance. 
In the following theorem summarizing our main result, the first given running 
time assumes approximate shortest path covers. In general it is NP-hard to 
compute the highway dimension~\cite{FeldmannFKP-highway-2015}, but it is unknown 
whether this problem is FPT. If this is the case and the running time is 
sufficiently small, this can be used as an oracle in our algorithm.

\begin{thm}\label{thm:alg}
For any graph $G$ with $n$ vertices and highway dimension $h$, there is an 
algorithm that computes a $3/2$-approximation to the \pname{$k$-Center} problem 
in time $2^{O(kh\log h)}\cdot n^{O(1)}$. If locally $h$-sparse shortest path 
covers are given by an oracle, the running time is $3^{kh}\cdot n^{O(1)}$.
\end{thm}

We leave open whether approximation factors better than $3/2$ can be obtained 
for the combined parameter~$(k,h)$. It was recently proved~\cite{k-center-hard} 
that \pname{$k$-Center} is W[1]-hard for this parameter~$(k,h)$, but no 
inapproximability is implied by this result. We note that a recent result by 
\citet{hd-Becker2017} obtains a \emph{fixed-parameter approximation scheme} for 
\pname{$k$-Center} in low highway dimension graphs, i.e., an algorithm computing 
a $(1+\eps)$-approximation in time $f(k,h,\eps)\cdot n^{O(1)}$ for any $\eps>0$. 
However, this result needs a more restrictive definition of the highway 
dimension than used in this paper. In particular, there are graphs that have 
bounded highway dimension due to \autoref{dfn:spc}, but for which the algorithm 
by \citet{hd-Becker2017} is not applicable (for a more detailed discussion on 
the relation between different definitions of highway dimension we refer 
to~\cite[Section~9]{FeldmannFKP-highway-2015}). Although we also leave open 
whether 
\mbox{$(2-\eps)$-FPA} algorithms exist for the parameter $h$ alone, we are able 
to prove that the techniques we use for \autoref{thm:alg} cannot omit using both 
$k$ and $h$ as parameters. To obtain a $(2-\eps)$-FPA algorithm with running 
time $f(h)\cdot n^{O(1)}$ for any function $f(\cdot)$ independent of~$n$, a lot 
more information of the input would need to be exploited than the algorithm of 
\autoref{thm:alg} does. To explain this, we now turn to the used techniques.

\subsection{Used techniques}

\begin{wrapfigure}[13]{R}{0.36\textwidth}
\vspace{-3mm}
\centering{\includegraphics[width=0.35\textwidth]{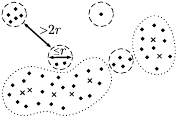}}
\caption{\label{fig:clusters} Clusters (dashed circles) are far from hubs 
(crosses). They have small diameter and are far from each other.} 
\end{wrapfigure}
A crucial observation for our algorithm is that at any scale $r$, a graph of low 
highway dimension is structured in the following way (see 
\autoref{fig:clusters}). We will prove that the vertices are either at distance 
at most $r$ from some hub, or they lie in clusters of diameter at most $r$ that 
are at distance more than $2r$ from each other. Hence, for the cost $\opt$ of 
the optimum \pname{$k$-Center} solution, at scale $r=\opt/2$ a center that 
resides in a cluster cannot cover any vertices of some other cluster. In this 
sense the clusters are ``independent'' of each other. At the same time we are 
able to bound the number of hubs of scale $\opt/2$ in terms of $k$ and the 
highway dimension. Roughly, this is comparable to graphs with small vertex 
cover, since the vertices that are not part of a vertex cover form an 
independent set. In this sense the highway dimension is a generalization of the 
vertex cover number (this is in fact the reason why computing the highway 
dimension is NP-hard~\cite{FeldmannFKP-highway-2015}).

At the same time the \pname{$k$-Center} problem is a generalization of the 
\pname{Dominating Set} problem. This problem is 
\mbox{W[2]-hard~\cite{downey2013fpt}}, which, as we will show, is also why 
\pname{$k$-Center} is W[2]-hard to approximate for parameter~$k$. However, 
\pname{Dominating Set} is FPT using the vertex cover number as the 
parameter~\cite{alber2002dom-planar}. This is one of the reasons why combining 
the two parameters $k$ and $h$ yields a $3/2$-FPA algorithm for 
\pname{$k$-Center}. In fact the similarity seems so striking at first that one 
is tempted to reduce the problem of finding a $3/2$-approximation for 
\pname{$k$-Center} on low highway dimension graphs to solving \pname{Dominating 
Set} on a graph of low vertex cover number. However, it is unclear how this can 
be made to work. Instead we devise an involved algorithm that is driven by the 
intuition that the two problems are similar. 

The algorithm will guess the cost $\opt$ of the optimum solution in order to 
exploit the structure of the graph given by the locally $h$-sparse shortest 
path 
cover for scale $r=\opt/2$. In particular, the shortest path covers of other 
scales do not need to be locally sparse in order for the algorithm to succeed. 
We will show that there are graphs for which \pname{$k$-Center} is not 
$(2-\eps)$-approximable, unless P=NP, and for which the shortest path cover for 
scale $\opt/2$ is locally $46$-sparse. Hence our techniques, which only 
consider the shortest path cover of scale $\opt/2$, cannot yield a 
$(2-\eps)$-FPA algorithm for parameter $h$. The catch is though that the 
reduction produces graphs which do not have locally sparse shortest path covers 
for scales significantly larger than~$\opt/2$. Hence a $(2-\eps)$-FPA algorithm 
for parameter $h$ might still exist. However, such an algorithm would have to 
take larger scales into account than just $\opt/2$, and as mentioned above, it 
would have to have at least doubly exponential running time in $h$.

Proving that no $(2-\eps)$-FPA algorithm for parameter $k$ exists for 
\pname{$k$-Center}, unless FPT=W[2], is straightforward given the original 
reduction of~\citet{hsu1979k-center-hard} from the W[2]-hard \pname{Dominating 
Set} problem. For parameter~$h$, however, we develop some more advanced 
techniques. For the reduction we show how to construct a graph of low highway 
dimension given a metric of low \emph{doubling dimension} (see 
\autoref{sec:hard} for a formal definition), so that distances between vertices 
are preserved by a $(1+\eps)$ factor. The doubling 
dimension~\cite{gupta2003doubling} is a parameter that captures the bounded 
volume growth of metrics, such as given by Euclidean and Manhattan distances. 
Since \pname{$k$-Center} is not $(2-\eps)$-approximable in 
$L_\infty$-metrics~\cite{feder1988metric-k-center}, unless P=NP, and these have 
constant doubling dimension, we are able to conclude that the hardness 
translates to graphs of highway dimension~$O(\log^2 n)$.

\subsection{Generalizations}

In addition to \pname{$k$-Center}, in \autoref{sec:general} we obtain similar 
positive results for two generalizations of the problem by appropriately 
modifying our techniques. For the \pname{weighted $k$-Center} problem, the 
vertices have integer weights and the objective is to choose centers of total 
weight at most $k$ to cover all vertices with balls of minimum radius. This 
problem is $3$-approximable~\cite{hochbaum1986bottleneck,Vazirani01book} and no 
better approximation factor is known. However, we are able to modify our 
techniques to obtain a $2$-FPA algorithm for the combined parameter~$(k,h)$.

An alternative way to define the \pname{$k$-Center} problem is in terms of 
finding a \emph{star cover} of size $k$ in a metric, where the cost of the 
solution is the longest of any star edge in the solution. More generally, in 
their seminal work \citet{hochbaum1986bottleneck} defined the 
\pname{$(k,\mc{F})$-Partition} problem. Here a family of (unweighted) graphs 
$\mc{F}$ is given and the aim is to partition the vertices of a metric into $k$ 
sets and connect the vertices of each set by a graph from the family~$\mc{F}$. 
The solution cost is measured by the ``bottleneck'', which is the longest 
distance between any two vertices of the metric that are connected by an edge in 
a graph from the family~$\mc{F}$. The case when $\mc{F}$ contains only stars is 
exactly the \pname{$k$-Center} problem, given the shortest-path metric as input. 
The \pname{$(k,\mc{F})$-Partition} problem is 
$2d$-approximable~\cite{hochbaum1986bottleneck}, where $d$ is the largest 
diameter of any graph in $\mc{F}$. We show that a $3\delta$-FPA algorithm for 
the combined parameter $(k,h)$ exists, where $\delta$ is the largest radius of 
any graph in $\mc{F}$. Hence for graph families in which $3\delta<2d$ this 
improves on the general algorithm by \citet{hochbaum1986bottleneck}. This is for 
example the case when $\mc{F}$ contains ``stars of paths'', i.e., stars for 
which each edge is replaced by a path of length at most $\delta$. The diameter 
of such a graph can be $2\delta$, while the radius is at most $\delta$, and 
hence~$3\delta<2d=4\delta$.

\subsection{Related work}

Given its applicability to various problems in transportation networks, but also 
in other contexts such as image processing and data-compression, the 
\pname{$k$-Center} problem has been extensively studied in the past. We only 
mention closely related results here, that were not mentioned before. For 
parameters clique-width and tree-width, \citet{katsikarelis2017structural} show 
that \pname{$k$-Center} is W[1]-hard, but they also give fixed-parameter 
approximation schemes for each of these parameters. For the tree-depth 
parameter, they show that the problem is FPT. For unweighted planar and map 
graphs the \pname{$k$-Center} problem is FPT~\cite{demaine2005planar-center} for 
the combined parameter~$(k,\opt)$, where $\opt$ is the cost of the optimum 
solution. Note though that $k$ and $\opt$ are somewhat opposing parameters in 
the sense that typically if $k$ is small then $\opt$ will be large, and vice 
versa. A very recent result~\cite{planar-EPTAS} gives an \emph{efficient} 
polynomial-time approximation scheme (EPTAS) for \pname{$k$-Center} on weighted 
planar graphs, which approximates both the optimum cost $\rho$ and the number of 
centers~$k$. That is, in time $f(\eps)\cdot n^{O(1)}$ the algorithm computes a 
$(1+\eps)$-approximation that uses at most $(1+\eps)k$ centers, for any 
$\eps>0$. Interestingly, this immediately implies a fixed-parameter 
approximation scheme for parameters $k$ and $\eps$ on weighted planar graphs: 
setting $\eps$ to a value smaller than $1/k$ forces the algorithm to compute a 
solution with at most $k$ centers (since $k$ is an integer), while the cost is 
within an $(1+\eps)$-factor of the optimum. \citet{marx2015optimal} prove that  
in planar graphs an optimum \pname{$k$-Center} solution can be computed in time 
$n^{O(\sqrt{k})}$. On the other hand, a recent result~\cite{k-center-hard} shows 
that \pname{$k$-Center} is W[1]-hard in planar graphs with constant doubling 
dimension, for the combined parameter $(k,h,t)$, where $h$ is the highway 
dimension and $t$ the treewidth of the input graph. Thus this problem remains 
hard, even when assuming that it abides to all aforementioned models of 
transportation networks at once. For any $L_q$ metric an $(1+\eps)$-FPA 
algorithm for the combined parameter $(k,\eps,D)$ can be 
obtained~\cite{agarwal2002clustering}, where $D$ is the dimension of the 
geometric space. This can also be generalized~\cite{k-center-hard} to an 
$(1+\eps)$-FPA algorithm for the combined parameter $(k,\eps,d)$, where $d$ is 
the doubling dimension.

\citet{abraham2010highway} introduced the highway dimension, and study it in 
several papers~\cite{abraham2010highway,abraham2010highway2,abraham2011vc}. 
Their main interest is in explaining the good performance of various 
shortest-path heuristics assuming low highway dimension. In~\cite{abraham2011vc} 
they show that a locally $O(h\log h)$-sparse shortest path cover can be computed 
in polynomial time for any scale if the highway dimension of the input graph is 
$h$, and each shortest path is unique. 
\citet{FeldmannFKP-highway-2015} consider computing 
approximations for various other problems that naturally arise in transportation 
networks. They show that quasi-polynomial time approximation schemes can be 
obtained for problems such as \pname{Travelling Salesman, Steiner Tree}, or 
\pname{Facility Location}, if the highway dimension is constant. For this 
however a more restrictive definition of the highway dimension than used here is 
needed (see~\cite[Section~9]{FeldmannFKP-highway-2015} for more details). The 
algorithms are obtained by probabilistically embedding a low highway dimension 
graph into a bounded treewidth graph while introducing arbitrarily small 
distortions of distances. Known algorithms to compute optimum solutions on low 
treewidth graphs then imply the approximation schemes. It is interesting to note 
that this approach does not work for the \pname{$k$-Center} problem since, in 
contrast to the above mentioned problems, its objective function is not linear 
in the edge lengths. As noted before however, a recent result by 
\citet{hd-Becker2017} obtains a fixed-parameter approximation scheme for 
\pname{$k$-Center} for combined parameter $(h,k,\eps)$ using a deterministic 
embedding, building on the results in~\cite{FeldmannFKP-highway-2015}. But 
again, for this the more restrictive definition of highway dimension also used 
in~\cite{FeldmannFKP-highway-2015} is needed. The only other theoretical results 
on the highway dimension that we are aware of at this point are by 
\citet{bauer2013search} and by \citet{kosowski2017beyond}. 
\citet{bauer2013search} show that for any graph $G$ there exist edge lengths 
such that the highway dimension is $\Omega(\mathrm{pw}(G)/\log n)$, where 
$\mathrm{pw}(G)$ is the pathwidth of~$G$. \citet{kosowski2017beyond} introduce 
the \emph{skeleton dimension} of a graph and compare it to the highway dimension 
in the context of shortest path heuristics.

\section{{\mdseries\pname{$k$-Center}} and highway dimension versus
{\mdseries\pname{Dominating Set}} and vertex covers}\label{sec:hd-vs-vc}

We begin by observing that the vertices of a low highway dimension graph are 
highly structured for any scale~$r$: the vertices that are far from any hub of a 
shortest path cover for scale $r$ are clustered into sets of small diameter and 
large inter-cluster distance (see \autoref{fig:clusters}). A similar observation 
was already made in~\cite{FeldmannFKP-highway-2015}, where clusters were called 
\emph{towns}. We need a slightly different definition of clusters than 
in~\cite{FeldmannFKP-highway-2015} however, which is why we do not use the 
same terminology here. For a set $S\subseteq V$ let $\dist_G(u,S)=\min_{v\in 
S}\dist_G(u,v)$ be the shortest-path distance from $u$ to the closest vertex in 
$S$.

\begin{dfn}\label{dfn:clusters}
Fix $r\in\mathbb{R}^+$ and a shortest path cover $\spc(r)\subseteq V$ for scale 
$r$ in a graph $G=(V,E)$. We call an inclusion-wise maximal set $T\subseteq 
\{v\in V\mid \dist_G(v,\spc(r))>r\}$ with $\dist_G(u,w)\leq r$ for all $u,w\in 
T$ a \emph{cluster}, and we denote the set of all clusters by $\mc{T}$. The 
\emph{non-cluster} vertices are those which are not contained in any cluster of 
$\mc{T}$.
\end{dfn}

Note that the set $\mc{T}$ is specific for the scale $r$ and the hub set 
$\spc(r)$. The following lemma summarizes the structure of the clusters and 
non-cluster vertices. Here we let $\dist_G(S,S')=\min_{v\in S}\dist_G(v,S')$ be 
the minimum distance between vertices of two sets $S$ and $S'$.

\begin{lem}\label{lem:clusters}
Let $\mc{T}$ be the cluster set for a scale $r$ and a shortest path cover 
$\spc(r)$. For each non-cluster vertex $v$, $\dist_G(v,\spc(r))\leq r$. The 
diameter of any cluster $T\in\mc{T}$ is at most~$r$, and $\dist_G(T,T')>2r$ for 
any distinct pair of clusters $T,T'\in\mc{T}$.
\end{lem}
\begin{proof}
The first two claims follow immediately from the definition of the clusters. 
For the third claim let $W=\{v\in V\mid \dist_G(v,\spc(r))>r\}$, such that any 
cluster $T\in\mc{T}$ is a subset of $W$. We first argue that there are no 
vertices $u,w\in W$ for which $\dist_G(u,w)\in (r,2r]$. If these existed, by 
\autoref{dfn:spc} there would be a hub $x\in\spc(r)$ hitting the shortest path 
between them. However, this path would have length 
$\dist_G(u,x)+\dist_G(w,x)>2r$ since $u$ and~$w$ are at distance more than~$r$ 
from $\spc(r)$, contradicting our assumption that $\dist_G(u,w)\leq 2r$. 

As a consequence, for any three vertices $u,v,w\in W$ with $\dist_G(u,v)\leq r$ 
and $\dist_G(v,w)\leq r$ we have $\dist_G(u,w)\leq\dist_G(u,v)+\dist_G(v,w)\leq 
2r$, and since we know that $\dist_G(u,w)\notin (r,2r]$, this implies that in 
fact $\dist_G(u,w)\leq r$. Hence the relation of being at distance at most $r$ 
in $W$ is transitive, and it is obviously also symmetric and reflexive, i.e., it 
is an equivalence relation on $W$. Moreover, any two vertices $u,w\in W$ that do 
not belong to the same equivalence class, i.e.~$\dist_G(u,w)>r$, must be at 
distance more than $2r$, as $\dist_G(u,w)\notin (r,2r]$. By 
\autoref{dfn:clusters} the clusters are exactly the equivalence classes of~$W$, 
and so $\dist_G(T,T')> 2r$ for any two distinct clusters $T,T'\in\mc{T}$.
\end{proof}

A \emph{vertex cover} $W$ is a subset of vertices such that every edge is 
incident to some vertex of $W$. In particular, if all edges have unit length, 
then a shortest path cover for scale $r=1/2$ is a vertex cover. Hence shortest 
path covers are generalizations of vertex covers.
A \emph{dominating set} $D$ is a subset of vertices such that every vertex is 
adjacent to some vertex of $D$. In a graph with unit edge lengths, a feasible 
\pname{$k$-Center} solution of cost~$1$ is a dominating set. In this sense the 
\pname{$k$-Center} problem is a generalization of the \pname{Dominating Set} 
problem, for which a dominating set of minimum size needs to be found. The 
\pname{Dominating Set} problem is $W[2]$-hard~\cite{downey2013fpt} for its 
canonical parameter (i.e., the size of the optimum dominating set), but it is 
FPT~\cite{alber2002dom-planar} for the parameter given by the \emph{vertex cover 
number}, which is the size of the smallest vertex cover of a given graph. As the 
following simple lemma shows, if $\opt$ is the cost of the optimum 
\pname{$k$-Center} solution, the number of hubs of the shortest path cover 
$\spc(\opt/2)$ is bounded in $k$ and the local sparsity of $\spc(\opt/2)$. Thus 
our setting generalizes the \pname{Dominating Set} problem on graphs with 
bounded vertex cover number. It is interesting to note that in contrast to the 
\pname{Dominating Set} problem being FPT for the vertex cover 
number~\cite{alber2002dom-planar}, our more general setting is 
W[1]-hard~\cite{k-center-hard}.

\begin{lem}\label{lem:bounded-hubs}
Let $\opt$ be the optimum cost of the \pname{$k$-Center} problem in a given 
instance $G$. If a shortest path cover $\spc(\opt/2)$ of $G$ for scale $\opt/2$ 
is locally $s$-sparse, then $|\spc(\opt/2)|\leq ks$.
\end{lem}
\begin{proof}
The optimum \pname{$k$-Center} solution covers the whole graph $G$ with $k$ 
balls of radius $\opt$ each. By \autoref{dfn:spc} there are at most $s$ hubs of 
$\spc(\opt/2)$ in each ball.
\end{proof}

We are able to exploit this intuition for our algorithm in \autoref{sec:alg}. On 
a high level, our algorithm follows the lines of the following simple procedure 
to solve \pname{Dominating Set} on graphs with bounded vertex cover number. As a 
subroutine we will solve an instance of the \pname{Set Cover} problem, for which 
a collection $\mc{S}\subseteq 2^U$ of subsets of a universe~$U$ is given 
together with a subset $U'\subseteq U$ of the universe.\footnote{Usually $U'=U$ 
but for convenience we define the problem slightly more general here.} A 
\emph{set cover} for $U'$ is a collection $\mc{S'}\subseteq\mc{S}$ of the sets 
in $\mc{S}$ covering $U'$, i.e., $\bigcup_{S\in\mc{S'}} S\supseteq U'$. The aim 
is to compute a minimum-sized set cover for the set system $(U',\mc{S})$. Given 
an input graph $G=(V,E)$ and a vertex cover $W\subseteq V$ of small size (which 
can, for instance, be an approximation), we perform the following three steps, 
in each of which we find a respective subset~$D_i$, $i\in\{1,2,3\}$, of the 
optimum dominating set $D\subseteq V$ of $G$. 
\begin{enumerate}
\item Guess the subset $D_1=W\cap D$ of vertices in the vertex cover $W$ that 
belong to the dominating set~$D$.
\item Since the vertices not in the vertex cover $W$ form an independent set, 
any vertex of $V\setminus W$, which is not adjacent to a vertex in $D_1$ must be 
in $D$. 
Thus we can let $D_2$ consist of all such vertices from~$V\setminus W$.
\item By our choice of $D_2$, if there are any vertices left in $V$ that are not 
adjacent to $D_1\cup D_2$, they must be in $W$. Furthermore these vertices must 
be adjacent to some vertices in $D$ contained in $V\setminus W$, by our choice 
of $D_1$. We can thus solve an instance of \pname{Set Cover}, where $U'$ is 
given by the subset of vertices in $W$ that are not adjacent to $D_1\cup D_2$, 
and the set system~$\mc{S}$ is given by the neighbourhoods of vertices in 
$V\setminus W$ restricted to $W$. The remaining set $D_3=D\setminus (D_1\cup 
D_2)$ consists of the vertices in $V\setminus W$ whose neighbourhoods form the 
smallest solution of this \pname{Set Cover} instance.
\end{enumerate}

For the first step of the above algorithm there are $2^{|W|}$ possible guesses 
for $D_1$. For each such guess, the second step can be performed in polynomial 
time. For the third step we need to solve \pname{Set Cover} for an instance with 
a small universe $U$. This can be done in $2^{|U|}\cdot (|U|+|\mc{S}|)^{O(1)}$ 
time using the algorithm of \citet{fomin2005exact}. Since in our case $U=W$ and 
$|\mc{S}|\leq |V\setminus W|$, this amounts to a running time of 
$2^{|W|}\cdot n^{O(1)}$. This \pname{Set Cover} algorithm is based on dynamic 
programming. During its execution the smallest set cover for every subset $U'$ 
of the universe $U$ is computed, and these optimum solutions are stored in a 
table. Therefore, instead of running an algorithm for \pname{Set Cover} for each 
guess of~$D_1$ in the third step above, we may run the algorithm of 
\citet{fomin2005exact} only once beforehand: we set the universe to all of $W$, 
and the set system will contain all neighbourhood sets of vertices in 
$V\setminus W$. This way the needed optimum solution for the corresponding 
subset $U'$ of~$U$ can be retrieved in constant time in the third step of our 
procedure. As we need to retrieve an optimum set cover for every guess of $D_1$, 
this improves the overall running time, which is now~$2^{|W|}\cdot n^{O(1)}$.

In our \pname{$k$-Center} algorithm we will use the same method of pre-computing 
a table containing all optimum \pname{Set Cover} solutions for subsets of a 
universe. We summarize the properties of the needed \pname{Set Cover} algorithm 
in the following.

\begin{thm}[\cite{fomin2005exact,pc-book}]\label{thm:set-cover-dp}
Given a set system $(U,\mc{S})$ we can compute a table $\mathbb{T}$, which for 
any subset $U'\subseteq U$ contains the smallest set cover for $(U',\mc{S})$ 
in the entry $\mathbb{T}[U']$. For any subset $U'\subseteq U$, the optimum set 
cover for $U'$ can be retrieved in constant time from $\mathbb{T}$, and 
$\mathbb{T}$ can be computed in $2^{|U|}\cdot (|U|+|\mc{S}|)^{O(1)}$ time.
\end{thm}

\section{The fixed-parameter approximation algorithm}\label{sec:alg}

We begin with a brief high-level description of the algorithm. As observed in 
\autoref{sec:hd-vs-vc}, we can think of solving \pname{$k$-Center} in a low 
highway dimension graph as a generalization of solving \pname{Dominating Set} in 
a graph with bounded vertex cover number. Our algorithm (see \autoref{alg_main}) 
is driven by this intuition. After guessing the optimum \pname{$k$-Center} 
cost~$\opt$ and computing $\spc(\opt/2)$ together with its cluster set~$\mc{T}$, 
we will see how the algorithm computes three approximate center sets $C_1$, 
$C_2$, and~$C_3$ (analogous to the three respective sets $D_1$, $D_2$, $D_3$ 
for \pname{Dominating Set}). For the first set~$C_1$ the algorithm guesses a 
subset of the hubs of $\spc(\opt/2)$ that are close to the optimum center set. 
This can be done in time exponential in $k$ and the local sparsity of the hub 
set, because there are at most that many hubs for scale~$\opt/2$ by 
\autoref{lem:bounded-hubs}. We will observe that by \autoref{lem:clusters} an 
optimum center lying in a cluster cannot cover any vertices that are part of 
another cluster. This makes it easy to determine a second set $C_2$ of 
approximate centers, each of which will lie in a cluster that must contain an 
optimum center. The third set of centers~$C_3$ will consist of cluster vertices 
that cover the remaining vertices not yet covered by $C_1$ and $C_2$. These 
remaining uncovered vertices will all be non-cluster vertices, and we find $C_3$ 
by solving a \pname{Set Cover} instance, similar to the third step in our 
procedure for \pname{Dominating Set}.

More concretely, consider an input graph $G=(V,E)$ with an optimum 
\pname{$k$-Center} solution $C^*$ of cost~$\opt$. In \autoref{alg:guess-opt1} to 
\autoref{alg:guess-opt2} of \autoref{alg_main} we try scales~$r$ in increasing 
order, to guess the correct value for which~$r=\opt/2$. For each guessed value 
of $r$ the algorithm computes a shortest path cover $\spc(r)$ together with its 
cluster set $\mc{T}$ in \autoref{alg:compute-spc}. By~\cite{abraham2011vc}, 
locally $O(h\log h)$-sparse shortest path covers are computable in polynomial 
time if the input graph has highway dimension~$h$. In \autoref{alg:sparsity} we 
therefore set $s$ to the bound of the local sparsity guaranteed 
in~\cite{abraham2011vc} (if locally $h$-sparse shortest path covers are given by 
an oracle, we may at this point set $s=h$). In order to keep the running time 
low, the algorithm checks that the number of hubs is not too large in 
\autoref{alg:skip}: since by \autoref{lem:bounded-hubs} we have 
$|\spc(\opt/2)|\leq ks$, we can dismiss any shortest path cover containing more 
hubs.

Assume that $r=\opt/2$ was found. In the following, for an index $i\in\{1,2,3\}$ 
we denote by $R^*_i=\bigcup_{v\in C^*_i} B_v(\opt)$ and $R_i=\bigcup_{v\in C_i} 
B_v(\frac{3}{2}\opt)$ the regions covered by some set of optimum centers 
$C^*_i\subseteq C^*$ (with balls of radius $\opt$) and approximate centers 
$C_i\subseteq V$ (with balls of radius $\frac{3}{2}\opt$), respectively. In 
\autoref{alg:guess-hubs} the algorithm guesses a minimum-sized set $H$ of hubs 
in $\spc(\opt/2)$, such that the balls of radius $\opt/2$ around hubs in $H$ 
cover all optimum non-cluster centers. That is, if $C^*_1\subseteq C^*$ denotes 
the set of optimum non-cluster centers, each of which is at distance at most 
$\opt/2$ from some hub in $\spc(\opt/2)$, then $C^*_1 \subseteq\bigcup_{v\in H} 
B_v(\opt/2)$ and $H\subseteq\spc(\opt/2)$ is a minimum-sized such set. We 
choose 
this set of hubs $H$ as the first set of centers $C_1$ for our approximate 
solution in \autoref{alg:set-C1}. Note that due to the minimality of $H$ we have 
$|C_1|\leq|C^*_1|$. Also $R^*_1\subseteq R_1$ since for any center in $C^*_1$ 
there is a center (i.e., a hub) at distance at most $\opt/2$ in~$C_1$.

The next step is to compute a set of centers so that all clusters of the 
cluster set $\mc{T}$ of $\spc(\opt/2)$ are covered. Some of the clusters are 
already covered by the first set of centers $C_1$, and thus in this step we 
want to take care of all remaining uncovered clusters, i.e., those contained in 
$\mc{U}=\{T\in\mc{T}\mid T\setminus R_1\neq\emptyset\}$. By the definition of 
$C^*_1$, any remaining optimum center in $C^*\setminus C^*_1$ must lie in a 
cluster. Furthermore, the distance between clusters of $\spc(\opt/2)$ is more 
than~$\opt$ by \autoref{lem:clusters}, so that a center of $C^*\setminus C^*_1$ 
in a cluster $T$ cannot cover any vertices of another cluster $T'\neq T$. Hence 
if we guessed $H$ correctly, we can be sure that each cluster $T\in\mc{U}$ must 
contain a center of $C^*\setminus C^*_1$. For each (remaining) cluster 
$T\in\mc{U}$ we thus pick an arbitrary vertex $v\in T$ in \autoref{alg:set-C2} 
and declare it a center of the second set $C_2$ for our approximate solution. 
Hence if the optimum set of centers for $\mc{U}$ is $C^*_2=\{v\in C^*\mid\exists 
T\in\mc{U}:v\in T\}$, we have $|C_2|\leq |C^*_2|$ (if some cluster of $\mc{U}$ 
contains more than one optimum center in order to cover different parts of the 
non-cluster vertices, $C^*_2$ may be larger than $C_2$). Moreover, since the 
diameter of each cluster is at most~$\opt/2$ by \autoref{lem:clusters}, we get 
$R^*_2\subseteq R_2$.

At this time we know that all clusters in $\mc{T}$ are covered by the region 
$R_1\cup R_2$. Hence if any uncovered vertices remain in $V\setminus(R_1\cup 
R_2)$ for our current approximate solution, they must be non-cluster vertices. 
Just as $C^*_2$, by our definition of $C^*_1$, every remaining optimum center in 
$C^*_3=C^*\setminus (C^*_1\cup C^*_2)$ lies in some cluster. Since 
$R^*_1\subseteq R_1$ and $R^*_2\subseteq R_2$, any remaining uncovered vertex of 
$V\setminus(R_1\cup R_2)$ must be in the region $R^*_3$ covered by centers in 
$C^*_3$. Next we show how to compute a set $C_3$ such that the region $R_3$ 
includes all remaining vertices of the graph and $|C_3|\leq|C^*_3|$. Note that 
the latter means that the number of centers in $C_1\cup C_2\cup C_3$ is at 
most~$k$, since $C^*_1$, $C^*_2$, and $C^*_3$ are disjoint. 

To control the size of $C_3$ we will compute the smallest number of centers that 
cover parts of~$R^*_3$ with balls of radius $\opt$. In particular, in 
\autoref{alg:guess-hubs2} we guess the set of hubs $H'\subseteq 
\spc(\opt/2)\setminus H$ that lie in the region~$R^*_3$ (note that we exclude 
hubs of $H$ from this set). We then compute a center set $C_3$ of minimum size 
such that $H'\subseteq \bigcup_{v\in C_3}B_v(\opt)$. For this we reduce the 
problem of computing centers covering $H'$ to the \pname{Set Cover} problem with 
fixed universe size, as shown in \autoref{alg:red1} to~\autoref{alg:red2}. This 
reduction is performed before entering the loops guessing $H$ and $H'$ to 
optimize the running time. The universe $U$ of the \pname{Set Cover} instance 
consists of all hubs in the shortest path cover $\spc(r)$, while the set system 
$\mc{S}$ of the instance is obtained by restricting the balls $B_v(\rho)$ of 
radius $\rho$ around cluster vertices $v$ to the hubs. By 
\autoref{thm:set-cover-dp} there is an algorithm that computes the optimum 
\pname{Set Cover} solution for every subset of the universe. This algorithm is 
called in \autoref{alg:red2} of \autoref{alg_main} to fill a lookup 
table~$\mathbb{T}$ with these optima. We can thus retrieve the optimum 
\pname{Set Cover} solution for the subset $H'\subseteq\spc(r)$ in 
\autoref{alg:set-C3}, and let $C_3$ contain each cluster vertex $v$ for which 
the set of hubs contained in the ball $B_v(\rho)$ is part of the optimum 
solution covering $H'$, which is stored in the entry $\mathbb{T}[H']$ of the 
table. As the next lemma shows, we obtain the required properties for $C_3$.

\begin{lem}\label{lem:C3}
Assume the algorithm guessed the correct scale $r=\opt/2$ and the correct sets 
$H$ and $H'$. The set $C_3=\{ v\in \bigcup_{T\in\mc{T}} T \mid 
B_v(\rho)\cap\spc(r)\in \mathbb{T}[H'] \}$ is of size at most~$|C^*_3|$ and 
$H'\subseteq \bigcup_{v\in C_3}B_v(\opt)$.
\end{lem}
\begin{proof}
The second property clearly follows since the sets $B_v(\rho)\cap\spc(r)$ in 
$\mathbb{T}[H']$ form a set cover for~$H'$, such that every hub in $H'$ is at 
distance at most $\rho$ from some $v\in C_3$. To see that $|C_3|\leq |C^*_3|$, 
it suffices to show that the vertices in $C^*_3$ correspond to a feasible 
\pname{Set Cover} solution for~$H'$. If~$H'$ was guessed correctly, this set 
contains only hubs in the region~$R^*_3$. As $R^*_3$ is covered by balls of 
radius $\rho$ around the centers in $C^*_3$, the union $\bigcup_{v\in 
C^*_3}B_v(\rho)\cap\spc(r)$ contains~$H'$. Moreover, these sets 
$B_v(\rho)\cap\spc(r)$ are contained in the set system $\mc{S}$, since all 
centers of $C^*_3$ are contained in clusters by definition of $C^*_1$. Thus the 
sets $B_v(\rho)\cap\spc(r)$ form a set cover for $H'$ in the 
instance~$(\spc(r),\mc{S})$.
\end{proof}

It remains to show that the three computed center sets $C_1$, $C_2$, and $C_3$ 
cover all vertices of $G$, which we do in the following lemma. In particular, 
the union $C_1\cup C_2\cup C_3$ will pass the feasibility test in 
\autoref{alg:feasible} of the algorithm.

\begin{algorithm}[b!]
\caption{FPA algorithm for \pname{$k$-Center} in low highway dimension graphs}
\label{alg_main}

\normalsize

\LinesNumbered
\DontPrintSemicolon
\SetAlgoVlined

\SetKwFunction{SC}{SetCoverDP}
\SetKwFunction{sort}{sort}
\SetKw{cont}{continue}

\KwIn{Graph $G=(V,E)$ of highway dimension $h$ with optimum \pname{$k$-Center} 
cost 
$\opt$}
\KwOut{\pname{$k$-Center} set $C$ of cost at most $\frac{3}{2}\opt$}

\BlankLine
$s\gets O(h\log h)$ \tcp{local sparsity of efficiently computable shortest 
path cover}\label{alg:sparsity}

\BlankLine
$\mathbb{A}\gets$ \sort{$\{\dist_G(u,v) \mid u,v\in V\}$} \tcp{sort distances 
and store them in array $\mathbb{A}$} \label{alg:guess-opt1}

\For(\tcp*[h]{consider distances in increasing order}){$i\gets 0$ \KwTo 
${n \choose 2} -1$}
{
$r\gets \mathbb{A}[i]/2$ \tcp{guess $r=\opt/2$} \label{alg:guess-opt2}

\BlankLine
Compute locally $s$-sparse $\spc(r)$ with cluster set $\mc{T}$ \;
\label{alg:compute-spc}
\lIf(\tcp*[h]{too many hubs means $r\neq\opt/2$})
{$|\spc(r)|>ks$}{\cont \label{alg:skip}}

\BlankLine
\tcp{prepare the \pname{Set Cover} lookup table}
$V(\mc{T})\gets \bigcup_{T\in\mc{T}} T$\; \label{alg:red1}
$\mc{S}\gets \bigcup_{v\in V(\mc{T})}\{B_v(\rho)\cap\spc(r)\}$ \tcp{the set 
system is given by hubs in balls of radius $\rho$ around cluster vertices}
$\mathbb{T}\gets$ \SC{$\spc(r),\mc{S}$} \tcp{lookup table $\mathbb{T}$ contains 
an optimum set cover for every subset of the universe $\spc(r)$} 
\label{alg:red2}

\BlankLine
\tcp{guess minimum-sized set of hubs covering non-cluster centers}
\ForEach{$H\subseteq\spc(r)$}{\label{alg:guess-hubs}
$C_1\gets H$ \tcp{these hubs form the 1st set of centers}\label{alg:set-C1}

\BlankLine
\tcp{cover all clusters not covered by balls around $H$:}
$R_1\gets \bigcup_{v\in C_1} B_v(3r)$ \tcp{the region covered so far}
$\mc{U}\gets \{T\in\mc{T}\mid T\setminus R_1\neq\emptyset\}$ \tcp{the clusters 
that still need to be covered}
$C_2\gets\emptyset$\; 
\ForEach{$T\in\mc{U}$}
{
$v\in T$ \tcp{select arbitrary vertex in $T$}\label{alg:set-C2}
$C_2\gets C_2\cup\{v\}$ \tcp{the 2nd set of centers}
}

\BlankLine
\tcp{cover rest of non-cluster vertices by reducing to \pname{Set Cover}:}
\ForEach(\tcp*[h]{guess hubs covered by centers 
in clusters}){$H'\subseteq\spc(r)\setminus H$}{\label{alg:guess-hubs2}
$C_3\gets \{ v\in V(\mc{T}) \mid B_v(\rho)\cap\spc(r)\in \mathbb{T}[H'] \}$ 
\label{alg:set-C3} \tcp{the 3rd set of centers is given by cluster vertices 
whose balls of radius $\rho$ cover $H'$}

\BlankLine
\tcp{check whether the solution is feasible:}
$C\gets C_1\cup C_2\cup C_3$\;
$R\gets \bigcup_{v\in C} B_v(3r)$ \tcp{the covered region}
\lIf(\tcp*[h]{a feasible solution was found}){$|C|\leq k$ \KwSty{and} $R=V$} 
{\KwRet{$C$} } \label{alg:feasible}
}
}
}
\end{algorithm}

\begin{lem}\label{lem:covered-region}
Assume the algorithm guessed the correct scale $r=\opt/2$ and the correct sets 
$H$ and $H'$. The approximate center sets $C_1$, $C_2$, and $C_3$ cover all 
vertices of $G$, i.e., $R_1\cup R_2 \cup R_3=V$.
\end{lem}
\begin{proof}
The proof is by contradiction: assume there is a $v\in V\setminus(R_1\cup R_2 
\cup R_3)$ that is not covered by the computed approximate center sets. The 
idea 
is to identify a hub $y\in\spc(\opt/2)$ on the shortest path between $v$ and an 
optimum center $w\in C^*$ covering $v$. We will show that this hub $y$ must 
however be in $H'$ and therefore $v$ is in fact in $R_3$, since $v$ also turns 
out to be close to $y$.

To show the existence of $y$, we begin by arguing that the closest hub 
$x\in\spc(\opt/2)$ to $v$ is neither in $H$ nor in $H'$. We know that each 
cluster of $\mc{T}$ is in~$R_1\cup R_2$, so that $v\notin R_1\cup R_2$ must be 
a non-cluster vertex. Thus by \autoref{lem:clusters}, $\dist_G(v,x)\leq\opt/2$. 
The region $R_1$ in particular contains all vertices that are at distance at 
most $\opt/2$ from any hub in~$H=C_1$. Since $v\notin R_1$ and 
$\dist_G(v,x)\leq\opt/2$, this means that $x\notin H$. From $v\notin R_3$ we 
can also conclude that $x\notin H'$ as follows. By \autoref{lem:C3}, $C_3$ 
covers all hubs of $H'$ with balls of radius $\opt$. Hence if $x\in H'$ then 
$v$ would be at distance at most $\frac{3}{2}\opt$ from a center of $C_3$, 
i.e., $v\in R_3$.

From $x\notin H\cup H'$ we can conclude the existence of $y$ as follows.
Consider an optimum center $w\in C^*$ that covers $v$, i.e., $v\in B_w(\opt)$. 
Recall that $R^*_1\subseteq R_1$ and~$R^*_2\subseteq R_2$. Since $v\notin 
R_1\cup R_2$, this means that $w$ is neither in $C^*_1$ nor in~$C^*_2$ so that 
$w\in C^*_3$. By definition of $H'$, any hub at distance at most $\opt$ from a 
center in $C^*_3$ is in $H'$, unless it is in $H$. Hence, as~$x\notin H\cup 
H'$, the distance between $x$ and $w$ must be more than $\opt$. Since 
$\dist_G(v,x)\leq \opt/2$, we get $\dist_G(v,w)>\opt/2$. We also know that 
$\dist_G(v,w)\leq\opt$, because $w$ covers~$v$. Hence the shortest path cover 
$\spc(\opt/2)$ must contain a hub~$y$ that lies on the shortest path 
between $v$ and~$w$. In particular, $\dist_G(v,y)\leq\opt$ and 
$\dist_G(y,w)\leq\opt$. Analogous to the argument used for $x$ above, $R_1$ in 
particular contains all vertices at distance at most $\opt$ from $H$, so that 
$y\notin H$ since $v\notin R_1$. However, then the distance bound for $y$ and 
$w$ yields $y\in H'$, as $w\in C^*_3$.

Since $C^*_1$ contains all non-cluster centers but $w\notin C^*_1$, by 
\autoref{lem:clusters} we get $\dist_G(y,w)>\opt/2$, which implies 
$\dist_G(v,y)<\opt/2$. But then $v$ is contained in the ball $B_{y}(\opt/2)$, 
which we know is part of the third region $R_3$ since $y\in H'$. This 
contradicts the assumption that $v$ was not covered by the approximate center 
set.
\end{proof}

Note that the proof of \autoref{lem:covered-region} does not imply that 
$R^*_3\subseteq R_3$, as was the case for $R_1$ and $R_2$. It suffices though to 
establish the correctness of the algorithm. Finally, we conclude the proof of 
\autoref{thm:alg} by analysing the runtime of the algorithm.

\begin{proof}[Proof of \autoref{thm:alg}]
By \autoref{lem:clusters} and \autoref{lem:C3}, if \autoref{alg_main} correctly 
guesses the cost $\opt$ and the two hub sets $H$ and $H'$, then $|C_1\cup 
C_2\cup C_3|\leq k$ and $R_1\cup R_2\cup R_3=V$. By \autoref{lem:bounded-hubs}, 
$|\spc(\opt/2)|\leq ks$ so that the correct value for $r$ will not be skipped in 
\autoref{alg:skip}. Hence by trying all possible values for $\opt$ in increasing 
order, \autoref{alg_main} will terminate with a feasible solution 
that covers all vertices with balls of radius $\frac{3}{2}\opt$, due to 
\autoref{alg:feasible}. To prove \autoref{thm:alg} it remains to bound the 
running time.

There are at most $n \choose 2$ possible values for $\opt$ that need to be 
tried by the outermost loop, one for every pair of vertices. Hence the only 
steps of \autoref{alg_main} that incur exponential running times are when 
guessing $H$ and~$H'$ and when filling the table $\mathbb{T}$ of the dynamic 
program for the \pname{Set Cover} problem. These steps are only performed for 
shortest path covers of size at most $ks$ due to \autoref{alg:skip}. Since we 
explicitly exclude the hubs in $H$ when choosing~$H'$, each hub of a shortest 
path cover can either be in $H$, in $H'$, or in none of them when trying all 
possibilities. Hence this gives $3^{ks}$ possible outcomes. Filling the table 
$\mathbb{T}$ takes $O(2^{ks}\cdot n^{O(1)})$ time according to 
\autoref{thm:set-cover-dp}, while retrieving an optimum solution for $H'$ in 
\autoref{alg:set-C3} can be done in constant time. Thus the total running time 
to compute a $3/2$-approximation is $O(3^{ks}\cdot n^{O(1)})$. If the input 
graph has highway dimension~$h$, \citet{abraham2011vc} show how to compute 
$O(\log h)$-approximations of shortest path covers in polynomial time if 
shortest paths have unique lengths. The latter can be assumed by slightly 
perturbing the edge lengths in such a way that any $3/2$-approximation in the 
perturbed instance also is a $3/2$-approximation in the original instance. 
Therefore we can set $s=O(h\log h)$ during the execution of our algorithm. If 
there is an oracle that gives locally $h$-sparse shortest path covers for each 
scale, then we can set $s=h$ instead. Thus the claimed running times follow.
\end{proof}

\section{Hardness results}\label{sec:hard}

We begin by observing that the original reduction of 
\citet{hsu1979k-center-hard} for \pname{$k$-Center} also implies that there are 
no $(2-\eps)$-FPA algorithms.

\begin{thm}\label{thm:hard-k}
It is W[2]-hard for parameter $k$ to compute a $(2-\eps)$-approximation to the 
\pname{$k$-Center} problem for any $\eps>0$.
\end{thm}
\begin{proof}[Proof (cf.~\cite{hsu1979k-center-hard,Vazirani01book})]
The reduction is from the \pname{Dominating Set} problem, which is 
W[2]-hard~\cite{downey2013fpt} for the standard parameter, i.e., the size of 
the smallest dominating set $D$ of the input graph $G$. The reduction simply 
introduces unit lengths for each edge of $G$, guesses the size of $D$, and sets 
$k=|D|$. Any feasible center set of cost $1$ corresponds to a dominating set, 
and vice versa. On the other hand, a center set has cost at least $2$ if and 
only if it is not a dominating set. Hence if the size of $D$ is guessed in 
increasing order starting from~$1$, $k$~must be equal to $|D|$ the first time a 
$(2-\eps)$-approximation of cost $1$ is obtained by an algorithm for 
\pname{$k$-Center}. 
By guessing the size of $D$ in increasing order, this would result in an 
$f(|D|)\cdot n^{O(1)}$ time algorithm to compute the optimum dominating set if 
there was a $(2-\eps)$-FPA algorithm for parameter $k$ for \pname{$k$-Center}.
\end{proof}

We now turn to proving that $(2-\eps)$-approximations are hard to compute on 
graphs with low highway dimension. For this we introduce a general reduction 
from low doubling metrics to low highway dimension graphs in the next lemma. A 
metric $(X,\dist_X)$ has \emph{doubling dimension} $d$ if for every 
$r\in\mathbb{R}^+$, each set $S\subseteq X$ of diameter $2r$ is the union of at 
most $2^d$ sets of diameter~$r$. The \emph{aspect ratio} $\alpha$ of a metric 
$(X,\dist_X)$ is the maximum distance between any two vertices of $X$ divided 
by the minimum distance, i.e., 
$\alpha=\max\{\frac{\dist_X(s,t)}{\dist_X(u,v)}\mid s,t,u,v\in X\land u\neq 
v\}$.

\begin{lem}\label{lem:dd-hd}
Given any metric $(X,\dist_X)$ with constant doubling dimension $d$ and aspect 
ratio $\alpha$, for any~$0<\eps<1$ there is a graph $G=(X,E)$ of highway 
dimension $O((\log(\alpha)/\eps)^d)$ on the same vertex set such that for 
all~$u,v\in X$,
$
\dist_X(u,v)\leq\dist_G(u,v)\leq (1+\eps)\dist_X(u,v).
$
Furthermore, $G$ can be computed in polynomial time from the metric.
\end{lem}
\begin{proof}
First off, by scaling we may assume w.l.o.g.\ that the minimum distance of the 
given metric is $\frac{2}{1+\eps}$. In particular this means that the maximum 
distance is $\frac{2\alpha}{1+\eps}$. A fundamental 
property~\cite{gupta2003doubling} of low doubling dimension metrics is that for 
any set of points $Y\subseteq X$ with aspect ratio $\alpha'$, the number of 
points $|Y|$ can be at most $2^{d\lceil\log_2\alpha'\rceil}$. The proof of this 
property is a simple recursive application of the doubling dimension definition. 
For each scale $2^i$ where \mbox{$i\in\{0,1,\ldots,\lceil\log_2 \alpha\rceil\}$} 
we will identify a sparse set~$Y_i$, which in any ball of radius~$2^{i+1}$ has 
aspect ratio~$O(\log(\alpha)/\eps)$. The idea is to use the vertices of $Y_i$ as 
hubs in a shortest path cover for scale~$2^i$, which then are locally sparse in 
any such ball. We will make sure that there is an index $i$ with a hub set $Y_i$ 
for any possible distance between vertex pairs in the resulting graph $G$. We 
need to make sure though that the shortest path for any pair of vertices passes 
through a corresponding hub of some~$Y_i$. We achieve this by adding edges 
between the hubs in~$Y_i$, which act as shortcuts. That is, the edges of $G$ 
will be slightly longer than the distances in the metric given by $\dist_X$, and 
we will make the distances shorter with increasing scales in order to guarantee 
that the shortest paths pass through corresponding hubs.

More concretely, consider any set $Z\subseteq X$ of vertices. A subset 
$Y\subseteq Z$ is a \emph{$\rho$-cover} of $Z$ if for every $v\in Z$ there is a 
$u\in Y$ such that $\dist_X(u,v)\leq\rho$, and $Y$ is a \emph{$\rho$-packing} of 
$Z$ if $\dist_X(u,v)>\rho$ for all distinct $u,v\in Y$. A \emph{$\rho$-net} of 
$Z$ is a set $Y\subseteq Z$ that is a $\rho$-cover and a $\rho$-packing of $Z$. 
It is easy to see that such a net can be computed greedily in $O(n^2)$ time. We 
will use sets $Y_i$ that form a \emph{hierarchy} $Y_i\subseteq Y_{i-1}$ of nets 
as hubs. In particular, $Y_0=X$ and $Y_i$ is a $\frac{\eps 2^{i-3}}{(1+\eps)^2 
L}$-net of $Y_{i-1}$ for each~$i\geq 1$, where $L=\lceil\log_2\alpha\rceil$ is 
the index of the largest scale. Note that due to the triangle inequality of the 
metric, each $Y_i$ is a \smash{$2\frac{\eps 2^{i-3}}{(1+\eps)^2 L}$-cover} 
of~$X$.

In $G$, for each $i$ we connect two vertices $u,v\in Y_i$ by an edge $uv$ of 
length $(1+\eps(1-i/L))\dist_X(u,v)$. If a vertex pair is contained in several 
sets $Y_i$ of different scales, we only add the shortest edge according to the 
above rule, i.e., the edge for the largest index $i$. Hence the distance in $G$ 
between any $u,v\in Y_i$ is at most $(1+\eps(1-i/L))\dist_X(u,v)$. In 
particular, $\dist_G(u,v)\leq(1+\eps)\dist_X(u,v)$ for any $u,v\in X$ since 
$X=Y_0$. Note also that $1+\eps(1-i/L)\geq 1$ and hence 
$\dist_X(u,v)\leq\dist_G(u,v)$.

To bound the highway dimension of $G$, consider any pair $u,v\in X$, and let 
$i\in\{0,1,\ldots,L\}$ be such that $\dist_G(u,v)\in(2^i,2^{i+1}]$. Recall that 
the minimum distance according to $\dist_X$ is $\frac{2}{1+\eps}>1$ (as 
$\eps<1$), while the maximum distance is $\frac{2\alpha}{1+\eps}$. Accordingly, 
in $G$ all distances lie in $(1,2\alpha]$ so the index $i$ exists. We show that 
the shortest path between $u$ and $v$ passes through a hub of $Y_i$. We do this 
by upper bounding $\dist_G(u,v)$ in terms of $\dist_X(u,v)$ using a path that 
contains vertices of~$Y_i$. Then we lower bound the length of any path that does 
not pass through $Y_i$ and show that it is longer than the shortest path. 

Let $x\in Y_i$ be the closest hub to $u$ and let $y\in Y_i$ be the closest hub 
to~$v$. We begin by determining some distance bounds for these vertices. Since 
$Y_i$ is a $2\frac{\eps 2^{i-3}}{(1+\eps)^2 L}$-cover of $X$ in the 
metric according to~$\dist_X$, the distances in $G$ from $u$ to $x$ and from 
$v$ to $y$ are at most \smash{$2(1+\eps)\frac{\eps 2^{i-3}}{(1+\eps)^2 
L}=\frac{\eps 2^{i-2}}{(1+\eps)L}$} each. It also means that 
\smash{$\dist_X(x,y)\leq\dist_X(u,v)+2\cdot\frac{\eps 2^{i-2}}{(1+\eps)^2 
L}$}, since we can get from $x$ to $y$ through $u$ and $v$ in the metric. We 
know that $\dist_G(u,v)>2^i$ and thus we have 
$\frac{2^i}{1+\eps}<\dist_X(u,v)$. Using these bounds we get
\[
\dist_G(u,v)&\leq\dist_G(u,x)+\dist_G(x,y)+\dist_G(y,v) \\
&\leq \left[1+\eps\left(1-\frac{i}{L}\right)\right]\dist_X(x,y) + 
2\cdot\frac{\eps 2^{i-2}}{(1+\eps)L} \\
&< \left[1+\eps\left(1-\frac{i}{L}\right)\right] 
\left(\dist_X(u,v)+2\cdot\frac{\eps 2^{i-2}}{(1+\eps)^2 L}\right) + 
\frac{\eps}{2L}\dist_X(u,v) \\
&< \left[1+\eps\left(1-\frac{i}{L}\right) + 
\left(1+\eps\left(1-\frac{i}{L}\right)\right)\frac{\eps}{2(1+\eps)L} 
+\frac{\eps}{2L} \right]\dist_X(u,v)\\
&\leq \left[1+\eps\left(1-\frac{i}{L}\right) + 
(1+\eps)\frac{\eps}{2(1+\eps)L} +\frac{\eps}{2L} \right]\dist_X(u,v)\\
&= \left[1+\eps\left(1-\frac{i-1}{L}\right)\right]\dist_X(u,v).
\]

We now show that every path $P$ that does not use any hub of $Y_i$ is longer 
than $\dist_G(u,v)$. Since the hub sets of different scales form a hierarchy, 
any hub of scale $2^j$ with $j>i$ is also a hub for scale~$2^i$. Hence if $P$ 
does not pass through any hub of $Y_i$, it also cannot contain any vertex of 
$Y_j$ where $j>i$. Thus, if $P=(w_0,\ldots,w_l)$ where $w_0=u$ and $w_l=v$, any 
edge $w_jw_{j+1}$ on $P$ will be of length at least 
$(1+\eps(1-(i-1)/L))\dist_X(w_j,w_{j+1})$. The sum 
$\sum_{j=0}^{l-1}\dist_X(w_j,w_{j+1})$ of all the distances in the metric over 
the path $P$ is an upper bound on $\dist_X(u,v)$, and thus the length of $P$ is 
at least $(1+\eps(1-(i-1)/L))\dist_X(u,v)$. Since the distance $\dist_G(u,v)$ 
is strictly smaller than this bound by the above calculations, the shortest 
path between $u$ and $v$ in~$G$ passes through some hub of $Y_i$.

To bound the highway dimension, for any $r>0$ we still need to bound the number 
of hubs that hit shortest paths of length in $(r,2r]$ in a ball $B$ of radius 
$2r$ in $G$. Since our hub sets form a hierarchy, we may consider all shortest 
paths longer than $r$: if $i$ is the index such that $r\in(2^i,2^{i+1}]$, all 
shortest paths of length more than $2^i$ are hit by hubs of $Y_i$ because 
$Y_j\subseteq Y_i$ for all $j>i$. In~$G$ the ball $B$ has a diameter of at most 
$4r$. Measured in the metric according to $\dist_X$ the set of vertices in~$B$ 
also has a diameter of at most $4r\leq 2^{i+3}$, since 
$\dist_X(u,v)\leq\dist_G(u,v)$ for any vertices $u,v\in X$. Because $Y_i$ is a 
$\frac{\eps 2^{i-3}}{(1+\eps)^2 L}$-packing in the metric, the aspect ratio of 
$Y_i\cap B$ is $\alpha'\leq 64(1+\eps)^2 L/\eps$. By the fundamental property 
of low doubling metrics~\cite{gupta2003doubling} mentioned above, there are at 
most $(128(1+\eps)^2 L/\eps)^d$ hubs in $Y_i\cap B$, which concludes the proof.
\end{proof}

\citet{feder1988metric-k-center} show that, for any $\eps>0$, it is NP-hard to 
compute a $(2-\eps)$-approximation for the \pname{$k$-Center} problem in 
two-dimensional $L_\infty$ metrics. In particular, the metric is induced by a 
grid graph with unit edge lengths, so that the aspect ratio is at most $n$. The 
doubling dimension of any such metric is $2$, since a vertex set of 
diameter~$2r$ (contained in a ``square'' of side-length~$2r$) can be covered by 
$4$ vertex sets of diameter $r$ (contained in ``squares'' of side-length~$r$). 
By the reduction given in \autoref{lem:dd-hd} we can thus obtain graphs of 
highway dimension $O(\log^2 n)$ for which computing $(2-\eps)$-approximations to 
\pname{$k$-Center} is NP-hard. The challenge remains is to push the highway 
dimension bound of this inapproximability result down to a constant. This would 
mean that no \mbox{$(2-\eps)$-FPA} algorithm for \pname{$k$-Center} exists if 
the parameter is the highway dimension $h$, unless~P=NP. However, we can still 
argue that assuming the exponential time hypothesis (ETH)~\cite{eth,eth-2}, any 
$(2-\eps)$-FPA algorithm for parameter $h$ must have doubly exponential running 
time. In particular, the above hardness result implies a polynomial-time 
reduction from \pname{SAT} to \pname{$k$-Center} on graphs of highway dimension 
$O(\log^2 n)$. That is, given a \pname{SAT} formula of size $N$, the reduction 
will produce a graph with $n=N^{O(1)}$ vertices and highway dimension 
$h=O(\log^2 n)=O(\log^2 N)$. Thus an algorithm computing a 
$(2-\eps)$-approximation to \pname{$k$-Center} in time $2^{2^{o(\sqrt 
h)}}\cdot n^{O(1)}$ would be able to decide \pname{SAT} in time $2^{o(N)}\cdot 
N^{O(1)}$. However, this would contradict ETH. Thus if a 
$(2-\eps)$-approximation 
algorithm for \pname{$k$-Center} with parameter $h$ exists, it is fair to assume 
that its running time dependence on $h$ must be extremely large. To summarize we 
obtain the following lower bounds.

\begin{crl}\label{crl:hard-h-plog}
For any constant~$\eps>0$ it is NP-hard to compute a $(2-\eps)$-approximation 
for the \pname{$k$-Center} problem on graphs of highway dimension $O(\log^2 n)$. 
Moreover, there is no $(2-\eps)$-FPA for \pname{$k$-Center} parameterized by 
the highway dimension $h$ with runtime $2^{2^{o(\sqrt h)}}\cdot n^{O(1)}$, 
unless ETH fails.
\end{crl}

The following lemma gives further evidence that obtaining a $(2-\eps)$-FPA 
algorithm for parameter~$h$ is hard. As argued below, it excludes the existence 
of such algorithms that only use shortest path covers of constant scales.

\begin{lem}\label{lem:hard-h-const}
For any~$\eps>0$ it is NP-hard to compute a $(2-\eps)$-approximation for the 
\pname{$k$-Center} problem on graphs for which on any scale $r>0$ there is a 
locally $(3\cdot 2^{2r}-2)$-sparse shortest path cover $\spc(r)$. Moreover, this 
is true for instances where the optimum cost $\opt$ is at most $4$.
\end{lem}
\begin{proof}
The reduction is similar to the one used for \autoref{thm:hard-k}, but reduces 
from the NP-hard \pname{Dominating Set} problem on cubic 
graphs~\cite{alimonti2000cubic-graphs}. To obtain an instance of 
\pname{$k$-Center}, again we simply introduce unit edge lengths, guess the size 
of the minimum dominating set $D$, and set $k=|D|$. In contrast to the 
reduction of \autoref{thm:hard-k} however, we will guess the size of $D$ in 
decreasing order starting from $n$. As before, any feasible center set of cost 
$1$ corresponds to a dominating set, and vice versa, while on the other hand, a 
center set has cost at least $2$ if and only if it is not a dominating set. 
Hence whenever $k$ is at least~$|D|$ a $(2-\eps)$-approximation for 
\pname{$k$-Center} must have cost $\rho=1$, and the cost is at least $2$ for 
smaller $k$. Therefore guessing the size of $D$ in decreasing order, $k$ is 
equal to $|D|$ the last time a $(2-\eps)$-approximation of cost $1$ is computed 
by an algorithm for \pname{$k$-Center}.

Consider the value $k=|D|-1$, i.e., the iteration at which we realize the size 
of~$D$. If the number of connected components of the input graph exceeds $k$, we 
know that there cannot be a dominating set of size $k$, and we can dismiss this 
value as a guess for the size of $D$ right away. Otherwise, there is a 
connected component with at least two vertices of $D$, since $|D|=k+1$. It is 
easy to see that removing one of these two vertices results in a center set of 
size $k$ with cost at most $4$. Hence we only need to call the 
$(2-\eps)$-approximation algorithm for \pname{$k$-Center} on instances where 
the optimum cost is $\rho\leq 4$.

It is easy to see that any ball with radius $2r$ around a vertex $v$ contains 
at most $3(\sum_{i=0}^{2r-1} 2^i)+1=3\cdot 2^{2r}-2$ vertices, due to the bound 
on the maximum degree. Hence any set of hubs is locally $(3\cdot 
2^{2r}-2)$-sparse, which concludes the proof.
\end{proof}

Consider a $(2-\eps)$-FPA algorithm for \pname{$k$-Center}, which only takes 
shortest path covers of constant scales into account, where the parameter is 
their sparseness. That is, the algorithm computes a $(2-\eps)$-approximation 
using hub sets $\spc(r)$ only for values $r\leq R$ for some $R\in O(1)$, and the 
parameter is a value $s$ such that $\spc(r)$ is locally $s$-sparse for every 
$r\leq R$. By \autoref{lem:hard-h-const} such an algorithm would imply that 
P=NP. Moreover this is true even if $R\in O(\opt)$. Hence if it is possible to 
beat the inapproximability barrier of $2$ using the local sparseness as a 
parameter, then such an algorithm would have to take large (non-constant) scales 
into account. Note that the running time of our $3/2$-FPA algorithm can in fact 
be bounded in terms of the local sparseness of $\spc(\opt/2)$ instead of the 
highway dimension. The instances produced by the reduction of 
\autoref{lem:hard-h-const} have shortest path covers that are locally 
$46$-sparse on scale $r=\opt/2\leq 2$. Thus we obtain the following corollary, 
which is a matching hardness lower bound to our algorithm.

\begin{crl}
For any~$\eps>0$ it is NP-hard to compute a $(2-\eps)$-approximation for the 
\mbox{\pname{$k$-Center}} problem on graphs for which on scale $r=\rho/2\leq 2$ 
there is a locally $46$-sparse shortest path cover~$\spc(r)$, where $\rho$ is 
the optimum \pname{$k$-Center} cost.
\end{crl}

From this corollary and \autoref{thm:hard-k}, we can conclude that our algorithm 
necessarily needs to combine the parameter $h$ with $k$ in order to achieve its 
approximation guarantee of $3/2$.

\section{Generalizations of the {\mdseries\pname{$k$-Center}} 
problem}\label{sec:general}

{\bf The {\mdseries\pname{weighted $k$-Center}} problem} is defined by giving 
each vertex $v\in V$ an integer weight $w(v)\in\mathbb{N}$. The aim is to find a 
set $C\subseteq V$ of centers such that their total weight is at most $k$, 
i.e., 
$\sum_{v\in C} w(v)\leq k$, and the maximum distance of any vertex to its 
closest center is minimized. \citet{hochbaum1986bottleneck} gave a 
$3$-approximation to the problem, and no better approximation factor is known. 
However, \autoref{alg_main} can be modified to obtain a $2$-FPA algorithm for 
\pname{weighted $k$-Center} for parameters $k$ and $h$ in graphs of highway 
dimension $h$.

For this, \autoref{alg_main} will again guess $r=\opt/2$ in 
\autoref{alg:guess-opt1} to \autoref{alg:guess-opt2}, where $\opt$ is the cost 
of an optimum solution. The three center sets $C_1, C_2, C_3$ will be chosen 
more carefully respecting the weights. In particular, instead of setting 
$C_1=H$ in \autoref{alg:set-C1}, for each hub $x\in H$ we pick a cheapest 
vertex in the ball $B_{x}(r)$ around $x$ to be a center of $C_1$, i.e., we pick 
a vertex from $\arg\min\{w(u)\mid u\in B_x(r)\}$. If $H$ was guessed correctly 
so that each non-cluster center $u\in C^*_1$ of the optimum solution $C^*$ has 
a hub of $H$ at distance at most $r$, then there is also a center $v$ in $C_1$ 
at distance at most $2r$ from~$u$. Hence a ball of radius $4r$ around $v$ will 
contain the ball of radius $2r$ around $u$, i.e., $R^*_1\subseteq \bigcup_{v\in 
C_1}B_{v}(4r)$. Furthermore, $w(v)\leq w(u)$ and hence the total weight of 
$C_1$ is at most that of $C^*_1$.

In \autoref{alg:set-C2} of \autoref{alg_main}, instead of picking an arbitrary 
vertex of the cluster~$T$, we will pick a vertex of $T$ with minimum weight. 
Since the choice of a vertex in $T$ was arbitrary before, we still have 
$R^*_2\subseteq R_2$. Additionally the total weight of $C_2$ is at most that of 
$C^*_2$ since each cluster of $\mc{U}$ contains a center of $C^*_2$ if $H$ was 
guessed correctly. To compute $C_3$ we will solve the \pname{weighted Set Cover} 
problem in \autoref{alg:red2}, where the weight of each set 
$B_v(\rho)\cap\spc(r)$ equals $w(v)$. This can easily be done by adapting the 
dynamic program of \citet{fomin2005exact} to respect weights of sets
(cf.~\cite{pc-book}). Hence the weight of the resulting center set $C_3$ is at 
most that of $C^*_3$, and balls of radius $3r$ around centers in $C_3$ still 
cover all remaining vertices, by the same arguments as in the proof of 
\autoref{lem:covered-region}.

In conclusion, the set of centers $C=C_1\cup C_2\cup C_3$ computed by the 
modified algorithm has a total weight at most that of $C^*$, and balls of radius 
$4r$ around the centers in $C$ cover all vertices. Since the weights are 
integers, there are at most $ks$ hubs for scale $\opt/2$ if $\spc(\opt/2)$ is 
locally $s$-sparse (cf.~\autoref{lem:bounded-hubs}). Therefore we obtain a 
$2$-FPA algorithm for the combined parameter $(k,h)$.

\medskip

{\bf For the {\mdseries\pname{$(k,\mc{F})$-Partition}} problem} a family 
$\mc{F}$ of unweighted graphs is given, such that for any $n\in\mathbb{N}$ there 
is a graph in $\mc{F}$ with exactly $n$ vertices. Given an input metric 
$(X,\dist_X)$ and a value $c\in\mathbb{R}^+$, the \emph{bottleneck graph} 
$H_X(c)$ on vertex set $X$ has an edge for every pair of vertices $u,v\in X$ 
with $\dist_X(u,v)\leq c$. For the \pname{$(k,\mc{F})$-Partition} problem the 
minimum cost $c$ needs to be found such that $X$ can be partitioned into $k$ 
sets $X_1,\ldots,X_k$, and there is a spanning subgraph $G\in\mc{F}$ in $H_X(c)$ 
on the vertex set $X_j$ for each $j\in\{1,\ldots,k\}$.

Note that if $\mc{F}=\{K_{1,i}\}_{i\geq 0}$, i.e., each graph in the family is a 
star, we have the \pname{$k$-Center} problem, and if $\mc{F}=\{K_i\}_{i\geq 1}$, 
i.e., each graph in the family is a clique, we have the so-called 
\pname{$k$-Clustering} problem. The \emph{eccentricity} of a vertex $v$ is the 
maximum distance from $v$ to any other vertex in terms of number of edges (i.e., 
measured by the \emph{hop-distance}). The \emph{diameter} of an unweighted graph 
$G$ is defined as the maximum eccentricity of any vertex in $G$. If the diameter 
of each $G\in\mc{F}$ is at most $d$, then a $2d$-approximation can be obtained 
for the \pname{$(k,\mc{F})$-Partition} problem~\cite{hochbaum1986bottleneck}. 

Let the \emph{radius} of a graph $G$ be the minimum eccentricity of any vertex 
in $G$. For shortest-path metrics induced by graphs of highway dimension $h$, we 
can obtain a $3\delta$-FPA algorithm for the combined parameter~$(k,h)$ for the 
\pname{$(k,\mc{F})$-Partition} problem, if every graph in the family $\mc{F}$ 
has radius at most $\delta$. Hence for graph families $\mc{F}$ for which 
$3\delta < 2d$, this improves on the $2d$-approximation 
by~\citet{hochbaum1986bottleneck}. This is, for example, the case when $\mc{F}$ 
contains ``stars of paths'', i.e., stars for which each edge is replaced by a 
path of length at most $\delta$. The diameter of such a graph can be $2\delta$, 
while the radius is at most $\delta$, and hence~$3\delta<2d=4\delta$.

To obtain our algorithm we reduce the \pname{$(k,\mc{F})$-Partition} problem to 
\pname{$k$-Center}. Note that if there is an optimum solution to 
\pname{$(k,\mc{F})$-Partition} with cost $\opt$, then there must be a solution 
of cost $\delta\opt$ for \mbox{\pname{$k$-Center}}: for each $X^*_j$ of the 
optimum partition for \pname{$(k,\mc{F})$-Partition}, place a center at a vertex 
$v$ of~$X^*_j$ that minimizes the eccentricity in the graph $G\in\mc{F}$ 
spanning $X^*_j$. Since every edge of $G$ has length at most $\opt$, the ball 
$B_{v}(\delta \opt)$ will contain~$G$. Computing a $3/2$-approximation to 
\pname{$k$-Center} using \autoref{alg_main}, we obtain a set $C$ of $k$ centers 
such that the closest center to any vertex is at distance at most~$3\delta 
\opt/2$. For each center~$v_j\in C$, $j\in\{1,\ldots, k\}$, consider the set of 
vertices $X_j$ for which $v_j$ is the closest center (including $v_j$ itself), 
breaking ties arbitrarily. The distance between any two vertices in $X_j$ is at 
most $3\delta \opt$. Hence the vertices of the bottleneck graph $H_X(3\delta 
\opt)$ can be partitioned into the sets $X_1,\ldots,X_k$ such that each $X_j$ is 
a clique in $H_X(3\delta \opt)$. Clearly this also means that each $X_j$ has 
some graph of $\mc{F}$ as a spanning subgraph in~$H_X(3\delta \opt)$. Thus we 
obtain a $3\delta$-FPA algorithm to the \pname{$(k,\mc{F})$-Partition} problem 
for metrics induced by low highway dimension graphs, with the same asymptotic 
running time as \autoref{alg_main}. 

Note that the reduction would not yield an improved approximation ratio if a 
$2$-approximation was used to solve \pname{$k$-Center} (which in many cases is 
the best achievable approximation ratio, as summarized in the introduction), 
since the radius is always at least half the diameter of a graph, i.e., a 
$2d$-approximation is already a $4\delta$-approximation.

\section{Open problems}

In this last section we summarize some problems left open by our results. The 
most pressing unanswered question concerns the approximability of 
\pname{$k$-Center} using only the highway dimension $h$ as a parameter. Even 
though we obtained some partial answers in \autoref{sec:hard}, these do not 
exclude the existence of a $(2-\eps)$-FPA for parameter $h$ alone. Also whether 
better approximation ratios than~$3/2$ can be obtained for the combined 
parameter $(k,h)$ remains open. In particular, the approximation scheme given by 
\citet{hd-Becker2017} for parameter $(k,h)$ using the more restrictive highway 
dimension definition as used in~\cite{FeldmannFKP-highway-2015}, makes this an 
appealing possibility. Even more intriguing would be a hardness result that
excludes approximation schemes for parameter $(k,h)$ using the more general 
\autoref{dfn:spc} for the highway dimension. This would imply that the 
difference between these definitions is more than just ``cosmetic''. We note at 
this point that all lower bound results of \autoref{sec:hard} are applicable to 
the highway dimension definition used 
in~\cite{hd-Becker2017,FeldmannFKP-highway-2015} (only the constants in the 
sparsity of the shortest path covers increase).

Another interesting open question concerns the computability of the highway 
dimension. In particular, obtaining better approximation ratios than $O(\log 
h)$, as given in~\cite{abraham2011vc}, would improve the running time of not 
only the \pname{$k$-Center} algorithm presented here, but also the algorithms 
given in~\cite{hd-Becker2017,FeldmannFKP-highway-2015}. This is true even if 
the running time is parameterized in the highway dimension~$h$. Hence an 
important question is whether computing the highway dimension is FPT for the 
canonical parameter $h$, or even whether an $o(\log h)$-FPA algorithm exists for 
this parameter.

Finally, it would also be interesting to see whether the techniques developed 
here (or in~\cite{hd-Becker2017,FeldmannFKP-highway-2015}) can be used for 
other variants of the \pname{$k$-Center} problem, such as, for instance, the 
\pname{$k$-Supplier} problem~\cite{hochbaum1986bottleneck}.

\bibliographystyle{plainnat}
\bibliography{papers}

\end{document}